\documentclass[11pt, a4paper]{article}
\usepackage[T1]{fontenc}
\usepackage{amsfonts}
\usepackage{amsmath}
\usepackage{amssymb}
\usepackage{amsthm}
\usepackage{bbm}
\usepackage{bm}
\usepackage{mathrsfs}
\usepackage{verbatim}
\usepackage{setspace}
\usepackage{enumitem}

\theoremstyle{plain}
\newtheorem{theorem}{Theorem}[section]
\newtheorem{proposition}[theorem]{Proposition}
\newtheorem{lemma}[theorem]{Lemma}
\newtheorem{corollary}[theorem]{Corollary}

\theoremstyle{definition}
\newtheorem{definition}[theorem]{Definition}
\newtheorem{remark}[theorem]{Remark}

\theoremstyle{remark}

\renewcommand{\labelenumi}{(\roman{enumi})}

\renewenvironment{thebibliography}[1]{%
\begin{oldthebibliography}{#1}%
\setlength{\baselineskip}{.9em}
\linespread{.9}
\small
\setlength{\parskip}{0ex}%
\setlength{\itemsep}{.1em}%
}%
{%
\end{oldthebibliography}%
}
\newcommand{\eps}{\varepsilon}
\newcommand{\N}{\mathbb{N}}

\newcommand{\R}{\mathbb{R}}

\newcommand{\cF}{\mathcal{F}}
\newcommand{\cN}{\mathcal{N}}

\newcommand{\cE}{\mathcal{E}}

\newcommand{\br}[1]{\langle #1 \rangle}

\DeclareMathOperator{\Var}{Var}

\newcommand{\sint}{\stackrel{\mbox{\tiny$\bullet$}}{}}

\numberwithin{equation}{section}

\begin{document}%

\title{\vspace{-0em} Small-Time Asymptotics of Option Prices\\ and First Absolute Moments}
\date{First version: June 11, 2010. This version: June 11, 2011.}
\author{
  Johannes Muhle-Karbe%
  \thanks{
  ETH Zurich, Department of Mathematics, R\"amistrasse 101, 8092 Zurich, Switzerland,
  email addresses \texttt{johannes.muhle-karbe@math.ethz.ch}, \texttt{marcel.nutz@math.ethz.ch}.
  }
  \and
  Marcel Nutz%
  \footnotemark[1]
}
\maketitle \vspace{-1em}

\begin{abstract}
We study the leading term in the small-time asymptotics of at-the-money call option prices when the stock price process $S$ follows a general martingale. This is equivalent to studying the first centered absolute moment of $S$. We show that if $S$ has a continuous part, the leading term is of order $\sqrt{T}$ in time $T$ and depends only on the initial value of the volatility. Furthermore, the term is linear in $T$ if and only if $S$ is of finite variation. The leading terms for pure-jump processes with infinite variation
are between these two cases; we obtain their exact form for stable-like small jumps.
To derive these results, we use a natural approximation of $S$
so that calculations are necessary only for the class of L\'evy processes.
\end{abstract}

{\small


}\\

\noindent \emph{Acknowledgements} We thank Martin Keller-Ressel, Sergey Nadtochiy and Mark Podolskij for discussions, and an anonymous referee for detailed comments. The first author was partially supported by the National Centre of Competence in Research ``Financial Valuation and Risk Management'' (NCCR FINRISK), Project D1 (Mathematical Methods in Financial Risk Management), of the Swiss National Science Foundation (SNF). The second author acknowledges financial support by Swiss National Science Foundation Grant PDFM2-120424/1.

\section{Introduction}

We consider the problem of option pricing in mathematical finance where the price of an option on a stock $S$ is calculated as the expectation under a risk-neutral measure. As usual we assume that the stock price is modeled directly under this measure and set the interest rate to zero. Therefore, our basic model consists of a c\`adl\`ag martingale $S$ on a filtered probability space
$(\Omega,\cF, (\cF_t)_{t\geq 0}, P)$ satisfying the usual hypotheses; we assume for simplicity that $\cF_0$ is trivial $P$-a.s. Our main interest concerns the small-time asymptotics of \emph{European call option prices},
\begin{equation}\label{eq:asymptoticIntro}
  E[(S_T-K)^+] \quad \mbox{for} \quad T\downarrow 0,
\end{equation}
where $x^+:=\max\{x,0\}$ and $K\in\R$ is the \emph{strike price} of the option. More precisely, we shall mostly be interested in the \emph{leading order} asymptotics for the \emph{at-the-money} case $K=S_0$.
Option price asymptotics are used in finance to find initial values for model calibration procedures
and also for model testing as explained below. From a probabilistic point of view one can note that,
up to a factor two, these are also the asymptotics of  the centered absolute first moment and, at least for continuous $S$, of the expected local time at the origin.

In applications $S$ is often specified via a stochastic differential equation driven by a Brownian motion $W$
and a Poisson random measure $N(dt,dx)$ with compensator $F(dx)dt$; i.e., $S$ is of the form
\begin{equation}\label{eq:canonRepresentationIto}
  S=S_0+\sigma\sint W + \kappa(x) \ast (N-F(dx)\,dt),
\end{equation}
where $\sigma=(\sigma_t)$ and $\kappa=(\kappa_t(x))$ are predictable integrands and we denote by $\sigma\sint W_u = \int_0^u \sigma_t\,dW_t$ the It\^o integral and by $\kappa(x) \ast (N-F(dx)\,dt)_u$ the integral $\int_0^u\int_{\R} \kappa_t(x)\,(N(dt,dx)-F(dx)\,dt)$ of random measures. In view of the applications, we shall adopt this representation for $S$, but we recall in Remark~\ref{re:representation} that this entails no essential loss of generality: every c\`adl\`ag martingale with absolutely continuous characteristics can be represented in the form~\eqref{eq:canonRepresentationIto}, and this includes all models of interest.

The \emph{main results} for the at-the-money asymptotics~\eqref{eq:asymptoticIntro} are obtained under certain right-continuity conditions on the mapping $t\mapsto (\sigma_t,\kappa_t)$. If we exclude the trivial case $S\equiv S_0$, the possible convergence rates range from $\sqrt{T}$ to $T$. If $\sigma_0\neq0$, i.e., in the presence of a Brownian component, the leading term is given by $(|\sigma_0|/\sqrt{2\pi}) \sqrt{T}$ irrespective of the jumps. On the other hand, the leading term is $CT$ if and only if $S$ is of finite variation, and then $C$ is given explicitly in terms of $\kappa_0$ and $F$. For pure-jump processes with infinite variation the rate may be anywhere between $\sqrt{T}$ and $T$ and it need not be a power of $T$. We consider a class of processes, containing most relevant examples, whose small jumps resemble the ones of an $\alpha$-stable L\'evy process, where $\alpha\in[1,2)$. For $\alpha>1$ the leading term is $C T^{1/\alpha}$ while for $\alpha=1$ we find $C T|\log T|$; the constants are given explicitly.

The \emph{basic idea} to obtain these results is to calculate the option price asymptotics for a simple model $Z$ which approximates $S$ in a suitable sense.  More precisely, we obtain a natural approximation by freezing the coefficients $\sigma$ and $\kappa$ in~\eqref{eq:canonRepresentationIto} at time $t=0$, namely
\begin{equation}\label{eq:LevyApproxIntro}
  Z=S_0 + \sigma_0 W +\kappa_0(x)*(N-F(dx)\,dt).
\end{equation}
Note that $\sigma_0$ and $\kappa_0(x)$ are deterministic since $\cF_0$ is trivial.
We show that the \emph{L\'evy process} $Z$ has the same leading order asymptotics as $S$ under mild regularity conditions. Therefore, an explicit treatment is necessary only in the L\'evy case, for which much ``finer'' arguments are possible.
We prove that we can pass with an error of order $O(T)$ from one pure-jump L\'evy process to another one when the small jumps have a similar behavior and use this to reduce even further to very particular L\'evy processes.

\vspace{-.5em}\paragraph{Literature.}
Due to their importance for model calibration and testing, small-time asymptotics of option prices have received considerable attention in recent years; see \cite{AlosLeonVives.07, BerestyckiEtAl.02, BerestyckiEtAl.04, CarrWu.03, Durrleman.08, Durrleman.10, FengFordeFouque.10, FordeFigueroaLopez.10, Forde.09, FordeJacquier.09, FordeJacquierLee.10, GatheralEtAl.09, MedvedevScaillet.07, Roper.08}. A survey of recent literature is given in the introduction of Forde et al.~\cite{FordeJacquierLee.10}. We shall review only the works most closely related to our study; in particular, we focus on the at-the-money case.
Here Carr and Wu~\cite{CarrWu.03} is an early reference with results in the spirit of ours. The authors obtain by partially heuristic arguments that the order of convergence for finite variation jumps is $O(\sqrt{T})$ in the presence of a Brownian component and $O(T)$ otherwise. This is in a general model including, e.g.,  exponential L\'evy processes; however, a boundedness assumption on the coefficients of the log price excludes the application to the Heston model, for example.
For the pure-jump case with infinite variation the authors mention that there is a range of possibilities for the order of convergence and they illustrate this by the so-called log stable model. Given option price data, the results are used to study whether the underlying process has jumps.

Durrleman~\cite{Durrleman.08} determines the rate $O(\sqrt{T})$ and the corresponding coefficient in a similar model, again with bounded coefficients and finite variation jumps. The result is stated in terms of the implied volatility, which is an alternative parametrization for option prices (see also Corollary~\ref{co:implied1}).
Forde~\cite{Forde.09} studies a class of continuous uncorrelated stochastic volatility models and computes explicitly the first two leading terms corresponding to the orders $T^{1/2}$ and $T^{3/2}$ while Forde et al.~\cite{FordeJacquierLee.10} obtain the same two coefficients for the Heston model with correlation. Forde and Figueroa-L\'{o}pez~\cite{FordeFigueroaLopez.10} determine the leading order term for the CGMY model. More generally, Tankov~\cite{Tankov.10} obtains several results similar to ours in the setting of exponential L\'evy models; detailed references are given below. Finally, for (arithmetic) L\'evy processes $S$, the at-the-money option price is related to the power variation of order one: if $T=\delta$ is the mesh size, $E[|S_T|] = E [n^{-1} \sum_{1\leq k\leq n} |S_{k\delta} -S_{(k-1)\delta}| ]$ for any $n\in\N$ by the i.i.d.\ property of the increments. Hence it is not surprising that we shall benefit from results of Jacod~\cite{Jacod.07} on asymptotic properties of power variations.

The present paper is organized as follows. Section~\ref{se:approx} contains the approximation result, which is stated for general martingales. Section~\ref{se:continuous} contains the analysis for the order~$\sqrt{T}$ and Section~\ref{se:discontinuous} the higher leading orders corresponding to the pure-jump case. An appendix contains some standard results about L\'evy processes that are used in the body of the text, often without further mention. We refer to the monograph of Jacod and Shiryaev~\cite{JacodShiryaev.03} for any unexplained notion or notation from stochastic calculus.

\section{Approximation of the Process $S$}\label{se:approx}

In this section we compare two martingales $S$ and $S'$ with different coefficients and study the distance
$|S_T-S'_T|$ in mean as $T\downarrow 0$. The main application to be used in the sequel is the case where $S'$ is the L\'evy approximation~\eqref{eq:LevyApproxIntro} of $S$. In that case, the assumption in the following result becomes a H\"older-type condition in mean for the coefficients $\sigma_t$ and $\kappa_t$ of $S$.

\begin{proposition}\label{pr:approx}
  Let $S$ be a martingale of the form~\eqref{eq:canonRepresentationIto} and $S'$ a martingale of the analogous form
  $S'=S'_0 + \sigma'\sint W +\kappa'(x)*(N-F(dx)dt)$ with $S'_0=S_0$. Let $\gamma\geq 0$.
  \begin{enumerate}[topsep=3pt, partopsep=0pt, itemsep=1pt,parsep=2pt]

  \item If $E [(\sigma_t-\sigma'_t)^2 ]=O(t^\gamma)$ and $E\big[\int_\mathbb{R} |\kappa_t(x)-\kappa'_t(x)|^2 F(dx)\big]=O(t^\gamma)$, then
      $E [|S_T-S'_T|^2 ]=O(T^{(1+\gamma)})$ and  $E[|S_T-S'_T|]=O(T^{(1+\gamma)/2})$. %

  \item Let $\beta\in [1,2]$. If $E\big[\int_\mathbb{R} |\kappa_t(x)-\kappa'_t(x)|^\beta F(dx)\big]=O(t^\gamma)$ and  $\sigma\equiv 0$, then $E[|S_T-S'_T|^\beta]=O(T^{(1+\gamma)})$ and $E[|S_T-S'_T|]=O(T^{(1+\gamma)/\beta})$. %

  \end{enumerate}
  The assertions remain valid if $O(\cdot)$ is replaced by $o(\cdot)$ throughout.
\end{proposition}

\begin{proof}
  We set $X:=S-S'$ and denote by $X=X^c+X^d$ the decomposition into the continuous and purely discontinuous martingale parts; note that
  \[
    X^c = (\sigma-\sigma')\sint W \quad \mbox{and}\quad X^d=(\kappa(x)-\kappa'(x))*(N-F(dx)dt).
  \]
  Fix $\beta\in[1,2]$ and $\gamma\geq0$. We use the Burkholder-Davis-Gundy inequality (e.g., Protter~\cite[Theorem IV.48]{Protter.05}) to obtain
  \begin{align}\label{eq:proofApproxBDG}
  E\big[|X_T|^\beta\big]
      & \leq 2^{\beta-1} \Big( E\big[|X^c_T|^\beta\big] + E \big[|X^d_T|^\beta\big]\Big)  \nonumber \\
      & \leq C_\beta \Big( E\big[\br{X^c,X^c}_T^{\beta/2}\big] + E\big[[X^d,X^d]_T^{\beta/2}\big] \Big)
  \end{align}
  for a universal constant $C_\beta$ depending only on $\beta$. We first treat the pure-jump case~(ii), then $X^c\equiv0$
  and we only need to estimate the second expectation in~\eqref{eq:proofApproxBDG}. Recall that for a real sequence $y=(y_n)$ the norms $\|y\|_{\ell^p}=(\sum_n |y_n|^p)^{1/p}$ satisfy $\|y\|_{\ell^p}\geq \|y\|_{\ell^{q}}$ for $1\leq p\leq q<\infty$. We apply this for
  $p=\beta$ and $q=2$ to obtain that%
  \[
    [X^d,X^d]_T^{\beta/2} = \bigg(\sum_{t \leq T} |\Delta X_t|^2 \bigg)^{\beta/2} \leq \sum_{t \leq T} |\Delta X|^\beta
    = |\kappa(x)-\kappa'(x)|^\beta \ast N_T,
  \]
  hence using the definition of the compensator and Fubini's theorem we have
  \begin{align*}
    E\big[[X^d,X^d]_T^{\beta/2}\big]
    &\leq E\big[|\kappa(x)-\kappa'(x)|^\beta \ast N_T\big] \\
    &= E\big[|\kappa(x)-\kappa'(x)|^\beta \ast (F(dx)dt)_T \big]\\
    &= \int_0^T E\bigg[\int_\mathbb{R} |\kappa_t(x)-\kappa'_t(x)|^\beta F(dx)\bigg] \,dt.
  \end{align*}
  By assumption, the integrand is of order $O(t^\gamma)$, hence the integral is of order $O(T^{(1+\gamma)})$ and the first assertion of~(ii) follows by~\eqref{eq:proofApproxBDG}. The second assertion then follows by Jensen's inequality.

  We now turn to the case~(i). Of course, the previous estimates hold in particular for $\beta=2$, so it remains to consider the continuous part in~\eqref{eq:proofApproxBDG}. For $\beta=2$ this is
  \[
    E\big[\br{X^c,X^c}_T\big] = E\bigg[\int_0^T (\sigma_t-\sigma'_t)^2 \,dt\bigg] = \int_0^T E[(\sigma_t-\sigma'_t)^2] \,dt
  \]
  and so the conclusion is obtained as before. Finally, we note that the proof remains valid if $O(\cdot)$ is replaced by $o(\cdot)$ throughout.
\end{proof}

We illustrate the use of Proposition~\ref{pr:approx} by two applications to the approximation of stochastic differential equations (SDEs).
For the sake of clarity, we do not strive for minimal conditions.

\begin{corollary}\label{cor:levysde}
  Let $f: \mathbb{R} \to \mathbb{R}$ be continuously differentiable with bounded derivative and let
  $L$ be a square-integrable L\'evy martingale. Then the SDE
  \[
    dS_t = f(S_{t-})\, dL_t, \quad S_0\in\R
  \]
  has a unique solution $S$ and the L\'evy process $Z_t=S_0+f(S_0)L_t$ satisfies
  $ E[|S_T-Z_T|]=O(T)$ as $T\downarrow 0$.
\end{corollary}

\begin{proof}
  We recall from \cite[Theorem V.67]{Protter.05} that the SDE has a unique strong solution $S$ and that $t \mapsto E[S_t^2]$ is locally bounded. The L\'evy process $L$ has a representation of the form $L=cW + x\ast (N-F(dx)dt)$ and then
  \begin{align*}
    S & = S_0 + cf(S_{-}) \sint W+(f(S_{-})x)*(N-F(dx)dt),\\
    Z & = S_0 + cf(S_{0}) \sint W+(f(S_{0})x)*(N-F(dx)dt).
  \end{align*}
  It suffices to verify the conditions of Proposition~\ref{pr:approx}(i) for $\gamma=1$ and $S'=Z$. Since $S_t=S_{t-}$ $P$-a.s.\ for each $t$,
  we have
  $E[(\sigma_t-\sigma'_t)^2 ]= c^2 E[|f(S_t)-f(S_0)|^2]$ and
  $E\big[\int_\mathbb{R} |\kappa_t(x)-\kappa'_t(x)|^2 F(dx)\big] = E[|f(S_t)-f(S_0)|^2] \int |x|^2 F(dx)$. The last integral is
  finite since $L$ is square-integrable (Lemma~\ref{le:LevyFacts}(vi)), and so it suffices to show that
  $E[|f(S_t)-f(S_0)|^2]=O(t)$. Now $f$ is Lipschitz-continuous by assumption, so it remains to prove that
  \[
    E[|S_t-S_0|^2]=O(t).
  \]
  For this, in turn, it suffices to verify the conditions of Proposition~\ref{pr:approx}(i) for
  $\gamma=0$ and $S'\equiv S_0$ (and hence $\sigma'\equiv0$ and $\kappa'\equiv0$). Indeed,
  we have that
  $E[|\sigma_t|^2] = c^2 E[|f(S_t)|^2]=O(1)$ as
  the linear growth of $f$ ensures that $E[f(S_t)^2]$ is again locally bounded like $E[S_t^2]$, and
  similarly we have that
  $E\big[\int_\mathbb{R} |\kappa_t(x)|^2 F(dx)\big]=E[f(S_t)^2] \int  |x|^2 F(dx)=O(1)$  as $t\downarrow 0$.
\end{proof}

\begin{remark}
  If the square-integrable L\'evy martingale $L$ does not have a Brownian component and if its L\'evy measure $F$ satisfies
  \[
    \int_\mathbb{R} |x|^\beta F(dx)<\infty
  \]
  for some $\beta\in [1,2]$, then the assertion in Corollary~\ref{cor:levysde} can be strengthened to
  $E[|S_T-Z_T|]=O(T^{2/\beta})$.
  In particular, this applies for $\beta=1$ when $L$ is of finite variation. %
  The proof is as above, using part~(ii) of Proposition~\ref{pr:approx} instead of part~(i).
\end{remark}

We give a second example where the coefficient of the SDE is not Lipschitz-continuous, as this sometimes occurs in stochastic volatility models.

\begin{corollary}\label{cor:continuoussv}
  Assume that $S$ solves the SDE
  \[
    dS_t=S_{t}\sqrt{v_{t-}}\;dW_t,\quad S_0\in\R,
  \]
  where $v\geq0$ is a c\`adl\`ag adapted process.
  If $t \mapsto E[v^{2+\eps}_t]$ and $t \mapsto E[S_t^{4+\eps}]$ are bounded in a neighborhood of zero for some $\eps>0$,
  then the L\'evy process
  $Z_t=S_0+S_0\sqrt{v_0}\, W_t$
  satisfies
  $E[|S_T-Z_T|]=o(\sqrt{T})$.

  In particular, this applies when $v$ is a square-root process, i.e., when $S$ is the Heston model.
\end{corollary}

\begin{proof}
  By Proposition~\ref{pr:approx}(i) applied with $\gamma=0$, it suffices to verify that $E[(S_t\sqrt{v_t}-S_0\sqrt{v_0})^2]=o(1)$; notice that by continuity of $W$ the SDE does not change if one replaces $v_-$ by $v$. In view of
  \[
    E[(S_t\sqrt{v_t}-S_0\sqrt{v_0})^2]=E[S_t^2v_t-S^2_0v_0]+2S_0\sqrt{v_0}\,E[S_0\sqrt{v_0}-S_t \sqrt{v_t}\,]
  \]
  it suffices to check that $E[S^2_t v_t]\to S^2_0 v_0$ and $E[S_t\sqrt{v_t}]\to S_0\sqrt{v_0}$ as $t \downarrow 0$. Since $S\sqrt{v}$ is right-continuous, this readily follows by the Cauchy-Schwarz inequality and uniform integrability. That the assumptions are satisfied for the Heston model follows from,
  e.g., Cox et al.~\cite[Section~3]{CoxIngersollRoss.85} and the proof of
  Andersen and Piterbarg~\cite[Proposition 3.1]{AndersenPiterbarg.07}.
\end{proof}

\section{Option Price of Order $\sqrt{T}$}\label{se:continuous}

The main idea in this section is to calculate the option price for $S$ from~\eqref{eq:canonRepresentationIto}
via the approximation
\begin{equation}\label{eq:LevyApproxBody}
  Z:=S_0 + \sigma_0 W +\kappa_0(x)*(N-F(dx)dt).
\end{equation}
We first have to ensure that this expression makes sense. %
Indeed, if $S$ is a martingale, it follows that
\begin{equation}\label{eq:integrabilityS}
  \int_\R |\kappa_t(x)| \wedge |\kappa_t(x)|^2\,F(dx)<\infty\quad P\otimes dt\mbox{-a.e.},
\end{equation}
but of course this may fail on the nullset $\{t=0\}$.
Hence we make the \emph{standing assumption} that $Z$ is well defined and integrable; i.e., that $\kappa_0(x)$ is Borel-measurable and satisfies
\begin{equation}\label{eq:integrabilityZ}
  \int_\R |\kappa_0(x)| \wedge |\kappa_0(x)|^2\,F(dx)<\infty.
\end{equation}
In any reasonable situation, one will be able to infer this condition from~\eqref{eq:integrabilityS}.

We can now prove our result for the at-the-money option price of order $\sqrt{T}$.
In fact, we describe the slightly more general situation of \emph{almost} at-the-money strikes by considering a deterministic strike function $T\mapsto K_T$ such that $K_T\to S_0$ as $T\downarrow 0$. The main observation is that the coefficient of order $\sqrt{T}$ depends only on the initial value of $\sigma$ and that the jumps are irrelevant at this order. We denote by $\cN$ the Gaussian distribution.

\begin{theorem}\label{th:callATM}
  Let $S$ be a martingale of the form~\eqref{eq:canonRepresentationIto} and assume that
  \[
    \lim_{t\downarrow 0} E[(\sigma_t-\sigma_0)^2]=0\quad \mbox{and} \quad \lim_{t\downarrow 0} E\bigg[\int_{\R} |\kappa_t(x)-\kappa_0(x)|^2\,F(dx)\bigg]=0.
  \]
  If $K_T=S_0+ \theta \sqrt{T} + o(\sqrt{T})$ for some $\theta\in\R$, then
  \begin{equation}\label{eq:callATMExpansion}
    E[(S_T-K_T)^+]= E[\cN(-\theta,\sigma_0^2)^+]\,\sqrt{T} + o(\sqrt{T})\quad \mbox{as } T\downarrow 0.
  \end{equation}
  In particular, for the at-time-money case $K\equiv S_0$ we have that
  \[
    E[(S_T-S_0)^+]= \frac{|\sigma_0|}{\sqrt{2\pi}} \sqrt{T} + o(\sqrt{T})\quad \mbox{as } T\downarrow 0.
  \]
\end{theorem}

\begin{remark}\label{rk:callATM}
  \begin{enumerate}[topsep=3pt, partopsep=0pt, itemsep=1pt,parsep=2pt]\renewcommand{\labelenumi}{(\alph{enumi})}
  \item The form $K_T=S_0+ \theta \sqrt{T} + o(\sqrt{T})$ chosen in the theorem is in fact the only relevant one. Indeed, if the convergence $K_T\to S_0$ is slower than $\sim C\sqrt{T}$, then the leading order asymptotics of $E[(S_T-K_T)^+]$ will simply be determined by $(S_0-K_T)^+$, and if it is faster, we find the same asymptotics as in the at-the-money case.
      As usual, we write $f(T)\sim g(T)$ if $f(T)/g(T)\to1$ as $T\downarrow 0$.
     One can also note that the constant $\theta$ satisfies
     $\theta=(S_0-K_T)/\sqrt{T} + o(1)$ and can therefore be interpreted as a \emph{degree of moneyness} for the option; a similar notion was previously used by Medvedev and Scaillet~\cite{MedvedevScaillet.07}.

  \item Even for the class of continuous martingales, $\sqrt{T}$ is the highest order in which the option price depends only on the initial volatility $\sigma_0$. Indeed, assume that this were also the case for the order $T^{1/2+\eps}$ and some $\eps>0$.
  Set $\sigma_t=\sqrt{(1+2\eps)\,t^{2\eps}}$ and define the martingale $S_t=\int_0^t \sigma_s\,dW_s$.
  Then $S$ is Gaussian with variance $\Var(S_T)=\int_0^T \sigma_t^2\,dt=T^{1+2\eps}$ and hence
  \[
    E[S_T^+]=\frac{1}{\sqrt{2\pi}} \,T^{1/2+\eps}. %
  \]
  This coefficient is of course different from the one for $S'\equiv0$, which is another martingale with $\sigma_0 =0$.
  \end{enumerate}
\end{remark}

Let us formulate the theorem once more, in the language preferred by practitioners. If $S > 0$, the \emph{implied volatility} $\sigma_{impl}(T) \in [0,\infty]$ of an at-the-money call option with maturity $T$ is defined as the solution of
\begin{equation}\label{eq:implied}
  E[(S_T-S_0)^+]=S_0 \Phi\left(\frac{1}{2}\sigma_{impl}(T)\sqrt{T}\right)-S_0 \Phi\left(-\frac{1}{2}\sigma_{impl}(T)\sqrt{T}\right),
\end{equation}
where $\Phi$ denotes the standard normal distribution function. This means that the option price coincides with its counterpart in a Black-Scholes model with volatility parameter $\sigma_{impl}(T)$. Note that $\sigma_{impl}(T)$ exits and is unique since $x \mapsto \Phi(x)-\Phi(-x)$ is strictly increasing and maps $[0,\infty]$ to $[0,1]$ and moreover $E[(S_T-S_0)^+] \in [0,S_0]$.
The following generalizes the result of~\cite{Durrleman.08} to unbounded coefficients and infinite variation jumps.

\begin{corollary}\label{co:implied1}
  Let $S > 0$. Under the conditions of Theorem~\ref{th:callATM} we have $\sigma_{impl}(T) \to |\sigma_0|/S_0$; i.e., the implied volatility converges to the spot volatility of the continuous martingale part of the log price.
\end{corollary}

\begin{proof}
  By
  the asymptotic properties of $\Phi$ (see, e.g., Abramowitz and Stegun~\cite[Chapter 7]{AbramowitzStegun.64}), we have that $\Phi(x)-\Phi(-x) \sim \sqrt{2/\pi}\,x$ for small $x$. Since $\sigma_{impl}(T)\sqrt{T} \downarrow 0$ by Theorem~\ref{th:callATM},
  we obtain that
  \[
    \Phi\left(\frac{1}{2}\sigma_{impl}(T)\sqrt{T}\right)-\Phi\left(-\frac{1}{2}\sigma_{impl}(T)\sqrt{T}\right)
     \sim \frac{\sigma_{impl}(T)}{\sqrt{2\pi}} \sqrt{T} .
  \]
  As Theorem~\ref{th:callATM} yields $E[(S_T-S_0)^+]/S_0 \sim \frac{|\sigma_0|/S_0}{\sqrt{2\pi}}\sqrt{T}$, the definition of
  $\sigma_{impl}(T)$ shows that $\sigma_{impl}(T) \to |\sigma_0|/S_0$.
\end{proof}

\begin{proof}[Proof of Theorem~\ref{th:callATM}]
  We may use the put-call parity to rewrite the call price in terms of an absolute moment: since $S$ is a martingale,
  \begin{align*}
    E[|S_T-K_T|]
    &=E[2 (S_T-K_T)^+ -S_T+K_T]\\
    &=2E[(S_T-K_T)^+]-S_0+K_T.
  \end{align*}
  From this we also see that we may assume $K_T=S_0+ \theta \sqrt{T}$ (i.e., that the $o\big(\sqrt{T}\big)$ part does not matter).
  By a translation we may also assume that $S_0=K_0=0$. (We exploit here that we are working with a general class of martingales $S$, not necessarily positive.)

  \vspace{.5em}\noindent\emph{Step 1: Continuous L\'evy case.} Assume first that $S$ is a continuous L\'evy process, i.e.,
  that $S_t=\sigma_0 W_t$. Then
  \[
    E[(S_T-\theta\sqrt{T})^+]=\sqrt{T}E[(\sigma_0W_1-\theta)^+]=\sqrt{T}E[\cN(-\theta,\sigma_0^2)^+].
  \]

  \vspace{.5em}\noindent\emph{Step 2: General L\'evy case.}
  Let $S$ be a L\'evy process (i.e., $S=Z$) and denote by $S=S^c+S^d$ its decomposition into the continuous and purely discontinuous martingale parts. In view of
  \[
    |S_T^c-K_T|-|S_T^d| \leq |S_T-K_T|\leq |S_T^c-K_T|+|S_T^d|
  \]
  and Step~1, it suffices to show that $E[|S^d_T|]$ is of order $o\big(\sqrt{T}\big)$. To relax the notation, let us assume that $S^d=S$.
  We can further decompose $S$ into  a martingale with bounded jumps, which is in particular square-integrable, and a compound Poisson process $X$ which is integrable due to~\eqref{eq:integrabilityZ}. One can check by direct calculation or by an application of Theorem~\ref{th:finitevariation} below that $E[|X_T|]=O(T)$. That is, we may even assume that $S$ is a square-integrable pure-jump L\'evy martingale. Then
  \[
    |S_T|/\sqrt{T} \to 0 \quad\mbox{in probability as }T\downarrow0;
  \]
  see, e.g.,~\cite[Lemma~4.1]{Jacod.07}. To conclude that $E[|S_T|]/\sqrt{T}\to 0$, it suffices to show the uniform integrability of $\{S_T/\sqrt{T}\}_{T>0}$. But
  this set is even bounded in $L^2(P)$ as $S$ is square-integrable; indeed, $\|S_T/\sqrt{T}\|^2_{L^2(P)}= E\big[[S,S]_1\big]$ due to the relation $E\big[S_T^2\big]= E\big[[S,S]_T\big]=T E\big[[S,S]_1\big]$.

  \vspace{.5em}\noindent\emph{Step 3: General case.}
  When $S$ is as in the theorem, we approximate $S$ by the L\'evy process
  $Z$ from~\eqref{eq:LevyApproxBody}.
  The assumptions of Proposition~\ref{pr:approx}(i) are satisfied for $S'=Z$ and the order $o(1)$, hence we obtain that
  $E[|S_T-Z_T|]$ is of order $o\big(\sqrt{T}\big)$. In view of Step~2 applied to $Z$ and
  \[
    E[|Z_T-K_T|] - E[|S_T-Z_T|] \leq E[|S_T-K_T|] \leq E[|Z_T-K_T|] + E[|S_T-Z_T|],
  \]
  this completes the proof.
\end{proof}

\section{Option Prices with Higher Leading Orders}\label{se:discontinuous}

In this section, we consider \emph{pure-jump} martingales
\begin{equation}\label{eq:purejump}
  S=S_0+\kappa(x)*(N-F(dx)dt),
\end{equation}
i.e., we set $\sigma\equiv0$ in~\eqref{eq:canonRepresentationIto}. Then the term of order $\sqrt{T}$ in Theorem~\ref{th:callATM} vanishes and the leading order is higher than $1/2$. We shall again use the approximation result from Section~\ref{se:approx} to reduce to the L\'evy case and consider
\begin{equation}\label{eq:approximation}
  Z=S_0+\kappa_0(x) * (N-F(dx)dt).
\end{equation}
However, this case is now more involved since the results depend on the properties of the L\'evy measure.
We recall the standing assumption~\eqref{eq:integrabilityZ} which ensures that $Z$ is well defined and integrable.

\subsection{Finite Variation}

We first treat the case when $S$ is of finite variation, which leads to the highest possible (nontrivial) convergence rate for the at-the-money option price.  Indeed, the following result shows that this class of price processes is characterized by the rate $O(T)$.

\begin{theorem}\label{th:finitevariation}
  Let $S$ be a pure-jump martingale of the form~\eqref{eq:purejump} and
  assume that $\lim_{t\downarrow 0} E\big[\int_{\mathbb{R}} |\kappa_t(x)-\kappa_0(x)|F(dx)\big]=0$. Then
  the following are equivalent:
  \begin{enumerate}[topsep=3pt, partopsep=0pt, itemsep=1pt,parsep=2pt]
    \item $S$ is of finite variation on $[0,T]$ for some $T>0$,
    \item $E[(S_T-S_0)^+]=O(T)$ as $T \downarrow 0$.
  \end{enumerate}
  In that case,
  \[
    E[(S_T-S_0)^+]=\frac{1}{2}C\,T+o(T) \quad \mbox{as } T \downarrow 0,
  \]
  where $C:=\int_\R |\kappa_0(x)|\,F(dx) + \big|\int_\R \kappa_0(x)\,F(dx)\big|$.
\end{theorem}

\begin{proof}
  As in the proof of Theorem~\ref{th:callATM} we may assume that $S_0=0$ and then we have $E[|S_T|]=2E[S_T^+]$.
  Let us first clarify the meaning of~(i) under the given conditions.
  The assumed convergence implies in particular that
  \begin{equation}\label{eq:proofFiniteT}
    E\bigg[\int_0^T\int_\R |\kappa_t(x)-\kappa_0(x)|\,F(dx)dt\bigg]<\infty\quad\mbox{for some }T>0.
  \end{equation}
  Due to this fact, the following are actually equivalent to (i):
  \begin{enumerate}[topsep=3pt, partopsep=0pt, itemsep=1pt,parsep=2pt]\renewcommand{\labelenumi}{(\alph{enumi})}
    \item $Z$ is of integrable variation,
    \item $S$ is of integrable variation on $[0,T]$ for some $T>0$.
  \end{enumerate}
  Indeed, assume (i). Then $S=S^1-S^2$ on $[0,T]$ for two increasing processes $S^1$ and $S^2$ with $S^1_0=S^2_0=0$. Since the jumps of the martingale $S$ are integrable, the \emph{positive} stopping time $\tau:=\inf\{0\leq t \leq T:\, S^1_t+S^2_t\geq1\}$ is such that $S$ is of integrable variation on $[\![0,\tau]\!]$. By \cite[II.1.33b]{JacodShiryaev.03} this implies that $E\big[\int_0^{\tau}\int_\R |\kappa_t(x)|\,F(dx)dt\big]<\infty$. Now
  we can use that the product $\tau\int_\R |\kappa_0(x)|\,F(dx)=  \int_0^{\tau}\int_\R |\kappa_0(x)|\,F(dx)dt$ is bounded by
  \[
    \int_0^{\tau}\int_\R |\kappa_t(x)-\kappa_0(x)|\,F(dx)dt +\int_0^{\tau}\int_\R |\kappa_t(x)|\,F(dx)dt<\infty\quad P\mbox{-a.s.}
  \]
  to conclude via Lemma~\ref{le:LevyFacts}(vii) that $Z$ is of integrable variation.
  Furthermore, (a) together with~\eqref{eq:proofFiniteT} yields
  $E\big[\int_0^T\int_\R |\kappa_t(x)|\,F(dx)dt\big]<\infty$, which is (b) by~\cite[II.1.33b]{JacodShiryaev.03}, and clearly (b) implies (i).

  \vspace{.5em}\noindent\emph{Step 1: L\'evy case}. We first assume that $S=Z$. Moreover, we replace (i) by the equivalent
  condition~(a).

  \noindent \emph{(ii) implies~(a):}
  We consider an increasing  sequence of continuous functions $f_n$ on $\R$ satisfying
  \[
    0\leq f_n(x)\leq |x|\wedge n,\quad f_n(x) = 0 \mbox{ for } |x|<1/n,\quad \lim_n f_n(x)=|x|
  \]
  for all $x\in\R$ and $n\geq1$. For each $n$ we have
  \[
    \liminf_{T\downarrow 0} \frac{1}{T} E[|Z_T|] \geq \liminf_{T\downarrow 0} \frac{1}{T} E[f_n(Z_T)]
    = \int_{\R} f_n(\kappa_0(x))\,F(dx)
  \]
  where the equality follows from Sato~\cite[Corollary 8.9]{Sato.99} since $f_n$ vanishes in a neighborhood of the origin (this holds for any L\'evy measure).
  By monotone convergence as $n\to\infty$ we obtain
  \[
    \liminf_{T\downarrow 0} \frac{1}{T} E[|Z_T|] \geq \int |\kappa_0(x)|\,F(dx).
  \]
  The left hand side is finite by assumption, hence $Z$ is of integrable variation.

  \noindent \emph{(a) implies~(ii):} If $Z$ is of integrable variation, its total variation process $\mathrm{Var}(Z)_t:=\int_0^t|dZ_s|$ is an integrable L\'evy subordinator and by Lemma~\ref{le:LevyFacts}(vii)
  \[
    \frac{1}{T} E[\mathrm{Var}(Z)_T]=E[\mathrm{Var}(Z)_1] =\int |\kappa_0(x)|\,F(dx) + \bigg|\int \kappa_0(x)\,F(dx)\bigg|= C.
  \]
  Since $|Z_T|\leq \Var(Z)_T$, we conclude that
  \[
    \limsup_{T\downarrow 0} \frac{1}{T} E[|Z_T|] \leq C.
  \]
  For the converse inequality we consider the function $g_n(x)=|x|\wedge n$ for $n\geq1$. Since $Z$ is of finite variation, we obtain for each $n$ that
  \[
    \liminf_{T\downarrow 0} \frac{1}{T} E[|Z_T|] \geq  \lim_{T\downarrow 0} \frac{1}{T} E[g_n(Z_T)] = \int g_n(\kappa_0(x))\,F(dx) + \bigg|\int \kappa_0(x)\,F(dx)\bigg|
  \]
  as a consequence of~\cite[Theorem~2.1(i)(c)]{Jacod.07} since the increments of $Z$ are i.i.d.\
  (see also~\cite[Equation~(5.8)]{Jacod.07}).
  Applying monotone convergence as $n\to\infty$ to the right hand side, we conclude that
  \[
    \liminf_{T\downarrow 0} \frac{1}{T} E[|Z_T|] \geq C.
  \]
  Hence we have proved the claimed convergence rate for the case $S=Z$.

  \vspace{.5em}\noindent\emph{Step 2: General case}.
  Under the stated assumption, Proposition~\ref{pr:approx}(ii) with $\beta=1$ and $\gamma=0$ yields that $E[|S_T-Z_T|]=o(T)$ as $T \downarrow 0$. Hence (ii) holds for $S$ if and only if it holds for $Z$ and so Step~1 yields the equivalence of~(i) and (ii). Moreover, the leading constant $C$ is the same as for $Z$ since the approximation error is of order $o(T)$.
\end{proof}

\begin{remark}
  \begin{enumerate}[topsep=3pt, partopsep=0pt, itemsep=1pt,parsep=2pt]\renewcommand{\labelenumi}{(\alph{enumi})}
  \item For exponential L\'evy processes, where $\kappa(x)=e^x-1$, the formula in (ii) was previously obtained by Tankov~\cite{Tankov.10}. We thank the referee for pointing out this reference.

  \item The constant $C$ has the following feature. For a given absolute moment $\mu:=\int |\kappa_0(x)|\,F(dx)$ of the L\'evy measure of $Z$, the value of $C$ may range from $\mu$ to 2$\mu$ depending on $\int \kappa_0(x)\,F(dx)$. In particular, $C$ is minimal if the jumps are symmetric and maximal if all jumps have the same sign.

  \item As in Corollary~\ref{co:implied1}, it follows from Theorem~\ref{th:finitevariation} that the implied volatility satisfies $\sigma_{impl}(T) \sim \sqrt{\pi/2}\,C/S_0 \; \sqrt{T}$ as $T \downarrow 0$.

  \item The following observation from Step~1 in the proof seems worth being recorded: if $Z$ is
    a L\'evy martingale of finite variation,
    the three functions $T\mapsto E[|Z_T|]$, $T\mapsto E[\sup_{t\leq T} |Z_t|]$ and
    $T\mapsto E[\Var(Z)_T]$ all converge to zero as $T\downarrow 0$ with the same leading term $CT$.
  \end{enumerate}
\end{remark}

\subsection{Infinite Variation}

We now turn to pure-jump processes with infinite variation. By the previous results we know that the leading order for the at-the-money option price has to be strictly between $\sqrt{T}$ and $T$; however, it need not be a power of $T$. The class of possible L\'evy measures at time zero is very rich and it is unclear how to compute the exact order in general. A look at existing financial models suggests to impose additional structure which will pin down an order of parametric form.
The base case is the $\alpha$-stable L\'evy measure which is given by
\[
  \nu(dx)=\frac{g(x)}{|x|^{1+\alpha}}\,dx,\quad g(x):=\beta_-1_{(-\infty,0)}(x)+\beta_+ 1_{(0,\infty)}(x)
\]
for $\alpha\in(0,2)$ and two nonnegative constants $\beta_+$ and $\beta_-$. More precisely, $\nu$ corresponds to a L\'evy process $L_t$ following a stable law with index of stability $\alpha$, skewness parameter $\beta=(\beta_+-\beta_-)/(\beta_+ + \beta_-)$, shift parameter $\mu=0$ and scale parameter $c=(\beta_+ +\beta_{-})^{1/\alpha}t^{1/\alpha}$.
In the nondegenerate case $\beta_++\beta_->0$ the process $L$ has infinite variation if and only if $\alpha\in[1,2)$.
For $\alpha\in(1,2)$ the first moment $E[|L_t|]$ exists whereas for $\alpha\in(0,1]$ this is not the case and in particular $L$ cannot be a martingale. On the other hand,  only the small jumps are relevant for the option price asymptotics of order smaller than $O(T)$, which leads us to the following definition.

\begin{definition}\label{def:alphaLike}
  Let $\alpha_+,\alpha_- \in (0,2)$. A L\'evy process is said to have \emph{$(\alpha_+,\alpha_-)$-stable-like small jumps} if its L\'evy measure $\nu$ is of the form
  \[
    \nu(dx)=\bigg(\frac{f(x)}{|x|^{1+\alpha_-}}1_{(-\infty,0)}(x) + \frac{f(x)}{|x|^{1+\alpha_+}}1_{(0,\infty)}(x)\bigg) \,dx
  \]
  for a Borel function $f\geq0$ whose left and right limits at zero,
  \[
    f_+:=\lim_{x \downarrow 0} f(x)\quad \mbox{and} \quad f_-:=\lim_{x \uparrow 0} f(x),
  \]
  exist and satisfy $f(x)-f_+=O(x)$ as $x \downarrow 0$ and $f(x)-f_{-}=O(x)$ as $x \uparrow 0$.
\end{definition}

This class includes most of the processes used in financial modeling; e.g., the tempered stable and in particular the CGMY processes, of which the variance gamma process is a special case, or the normal inverse Gaussian process. We refer to Cont and Tankov~\cite[Section~4.5]{ContTankov.04} for more information on these models. We observe that if the driving L\'evy process of an SDE as in Corollary~\ref{cor:levysde} is chosen from this class, then the process $Z$ defined in the same corollary is again of the same type; merely the constants $f_+$ and $f_-$ change.

Since $\alpha_+\vee\alpha_-\in(0,1)$ implies that the jumps are of finite variation (cf.~Theorem~\ref{th:finitevariation}), we are interested here only in the case $\alpha_+\vee\alpha_-\in[1,2)$.
In that case we shall see that the larger of the values $\alpha_+$ and $\alpha_-$ determines the leading order.

The statement of our main result requires the following constants. For $f_+,f_-\geq0$ and $\alpha\in(1,2)$ we set $C(\alpha,0,0):=0$ and if $f_+ +f_->0$ then
\begin{align*}
  C(\alpha,f_+,f_{-}):=&\,\frac{2}{\pi}(f_+ + f_{-})^{\frac{1}{\alpha}}\;\Gamma\left(1-\frac{1}{\alpha}\right)\left[1+\left(\frac{f_+ - f_{-}}{f_+ + f_{-}}\right)^2 \tan^2\left(\frac{\alpha \pi}{2}\right)\right]^{\frac{1}{2\alpha}}\\
  & \times  \cos\left(\frac{1}{\alpha}\arctan\left(\frac{f_+ - f_{-}}{f_+ + f_{-}}\tan\left(\frac{\alpha\pi}{2}\right)\right)\right),
\end{align*}
where $\Gamma$ denotes the usual Gamma function. For $\alpha_+,\alpha_-\in(0,2)$ such that $\alpha_+\vee\alpha_-\in(1,2)$ we then define
\[
  C(\alpha_+,\alpha_-,f_+,f_-):=
  \begin{cases}
  C(\alpha_+,f_+,f_{-}), & \alpha_+=\alpha_-\,, \\
  C(\alpha_+,f_+,0), & \alpha_+>\alpha_-\,,\\
  C(\alpha_-,0,f_-), & \alpha_+<\alpha_-\,.
  \end{cases}
\]

\begin{theorem}\label{thm:jumps}
  Let $S$ be a pure-jump martingale of the form~\eqref{eq:purejump} such that the L\'evy process $Z$ from~\eqref{eq:approximation} has $(\alpha_+,\alpha_-)$-stable-like small jumps and let $\alpha:=\alpha_+\vee\alpha_-$. Assume that there exist $\beta \in [1,2]$ and $\gamma\geq0$ such that
  \[
    E\bigg[\int_{\mathbb{R}}|\kappa_t(x)-\kappa_0(x)|^\beta\, F(dx)\bigg]=o(t^\gamma) \quad \mbox{and}\quad \frac{\beta}{1+\gamma}\leq\alpha.
  \]
  \begin{enumerate}[topsep=3pt, partopsep=0pt, itemsep=1pt,parsep=2pt]
    \item If $\alpha\in(1,2)$, then
      \[
        E[(S_T-S_0)^+]= \frac{1}{2}C\big(\alpha_+,\alpha_-,f_+,f_{-}\big)\,T^{1/\alpha}+o(T^{1/\alpha}) \quad \mbox{as } T \downarrow 0.
      \]

    \item If $\alpha_+=\alpha_-=1$ and $f_+=f_-$, then
      \[
        E[(S_T-S_0)^+]= \frac{1}{2}\big(f_+ + f_{-}\big)\,T|\log T| +o(T|\log T|) \quad \mbox{as } T \downarrow 0.
      \]
  \end{enumerate}
\end{theorem}

\begin{remark}\label{rk:afterJumpsThm}
  \begin{enumerate}[topsep=3pt, partopsep=0pt, itemsep=1pt,parsep=2pt]\renewcommand{\labelenumi}{(\alph{enumi})}
  \item For exponential L\'evy processes, a result similar to part (i) is obtained in~\cite{Tankov.10}. There, the condition on the jumps in formulated in a slightly different way and hence our theorem is not a strict generalization. More specifically, in \cite{Tankov.10}, the author deals with L\'evy processes whose characteristic function resembles the one of a stable process. On the other hand, we consider processes whose jump measures
resemble the jump measure of a stable process around the origin, which seems more convenient beyond L\'evy models. For a stable process, one readily verifies that the two expressions for the small-time limit indeed coincide, once the relationship between the characteristic function and the jump measure (cf., e.g., \cite[Chapter I.1]{SamorodnitskyTaqqu.94}) is taken into account.

 The special case of a CGMY model is also treated in~\cite{FordeFigueroaLopez.10}.

  \item In part (i), the continuity assumption on $\kappa_t(x)$ is satisfied in particular if $E[\int_{\mathbb{R}}|\kappa_t(x)-\kappa_0(x)|^2F(dx)]=O(t)$, as $\alpha>1$. For example, this holds in the setting of the L\'evy-driven SDE of Corollary~\ref{cor:levysde} (see the proof of that result).

  \item In the limit $\alpha\uparrow2$, which corresponds to the Brownian case, we obtain the order $\sqrt{T}$ as in Theorem~\ref{th:callATM}. On the other hand, the limit $\alpha\downarrow1$ in part~(i) does not yield the order obtained in part~(ii) and the leading constants explode since $\lim_{a\downarrow 0}\Gamma(a)=+\infty$.

  \item The theorem still holds true if the regularity of $f$ in Definition~\ref{def:alphaLike} is weakened as follows: Instead of
  $f(x)-f_\pm=O(|x|)$, it is sufficient to have $f(x)-f_\pm=O(|x|^\varrho)$ for some $\varrho>\alpha/2$. The proof is identical.

  \item As in Corollary~\ref{co:implied1} we can deduce that the implied volatility satisfies
  $\sigma_{impl}(T) \sim \sqrt{\pi/2}\,C(\alpha_+,\alpha_{-},f_+,f_{-})/S_0\; T^{1/\alpha - 1/2}$ in the setting of
  part (i) and $\sigma_{impl}(T) \sim \sqrt{\pi/2}\,(f_++f_{-})/S_0\; \sqrt{T}|\log T|$ in the setting of part~(ii)
   of Theorem~\ref{thm:jumps}.

  \item The following consequence of a result due to Luschgy and Pag\`es (see \cite[Theorem~3]{LuschgyPages.08}) complements part~(ii): If $S$ is a $(1,1)$-stable-like L\'evy process, possibly with $f_+\neq f_-$, then $E[(S_T-S_0)^+]=O(T|\log T|)$ still holds. However, their method only yields an upper bound and not that the leading order is indeed $\sim CT|\log T|$. As in the proof below we can infer that the same bound holds if $S$ is not a L\'evy process but satisfies the assumptions of part (ii) excluding
      $f_+= f_-$.
  \end{enumerate}
\end{remark}

\subsubsection{Proof of Theorem~\ref{thm:jumps}}

The plan for the proof of Theorem~\ref{thm:jumps} is as follows. We shall again reduce from the general martingale to the L\'evy case by Proposition~\ref{pr:approx}. Since the first absolute moment is known for stable processes, the main step for part~(i) will be to estimate the error made when replacing a stable-like L\'evy process by a true stable process with index $\alpha=\alpha_+\vee\alpha_->1$. For part (ii) the situation is slightly different as the $1$-stable process fails to have a first moment; in this case we shall instead use the normal inverse Gaussian as a reference process.

In a first step we show more generally that we can pass with an error of order $O(T)$ from one L\'evy process to another one when the small jumps have a similar behavior.

\begin{lemma}\label{lem:approxlevy}
  Let $L$ and $L'$ be pure-jump L\'evy martingales with L\'evy measures $\nu$ and $\nu'$, respectively.
  Suppose that for some $\delta>0$, the Radon-Nikodym derivative
  \[
    \psi(x)=\frac{d\big(\nu'(dx)|_{[-\delta,\delta]}\big)}{d\big(\nu(dx)|_{[-\delta,\delta]}\big)}
  \]
  of the measures restricted to $[-\delta,\delta]$ exists and that $\int_{-\delta}^\delta |\psi(x)-1|^2\, \nu(dx)<\infty$. Then
  $E[|L_T|]=E[|L'_T|]+O(T)$ as $T\downarrow 0$.
\end{lemma}

\begin{proof}
  Let $L=x*(\mu-\nu(dx)dt)$ be the canonical representation of $L$; here $\mu$ is the random measure associated with the jumps of $L$. We decompose $L$ into $L=L^{\leq\delta} + L^{> \delta}$,
  where the L\'evy process
  \[
    L^{\leq\delta}:=x1_{\{|x|\leq\delta\}}*(\mu-\nu(dx)dt)
  \]
  is obtained by truncating the jumps at magnitude $\delta$. Then $L^{> \delta}=L-L^{\leq\delta}$ is of finite variation and Theorem~\ref{th:finitevariation} yields
  \begin{equation*}%
    E[|L^{> \delta}_T|]=O(T) \quad \mbox{as } T \downarrow 0.
  \end{equation*}
  Hence we may assume that $L=L^{\leq\delta}$, i.e., that the jumps are bounded by $\delta$ in absolute value, or equivalently that $\nu$ is concentrated on $[-\delta,\delta]$.
  The integrability assumption on $\psi$ ensures that $Y:=(\psi(x)-1)*(\mu-\nu(dx)dt)$ is a square-integrable L\'evy martingale; moreover, $\psi\geq0$ implies that $\Delta Y \geq -1$.
  Hence the stochastic exponential
  \[
    D:=\cE(Y)=\cE\Big((\psi(x)-1)*(\mu-\nu(dx)dt)\Big)
  \]
  is a nonnegative square-integrable martingale (cf.~Lemma~\ref{le:LevyFacts}(x)). We define the probability measure $Q \ll P$ on $\cF_1$ by $dQ/dP=D_1$. Then the Girsanov-Jacod-M\'emin theorem (cf.~\cite[III.3.24]{JacodShiryaev.03}) yields that under $Q$, the process $(L_t)_{0 \leq t \leq 1}$ is L\'evy with triplet $(b_Q,0,\nu_Q)$ relative to the truncation function $h(x)=x$, where $\nu_Q(dx)=\psi(x)\,\nu(dx)=1_{[-\delta,\delta]}\nu'(dx)$ and
  \[
    b_Q=\int_{-\delta}^\delta x(\psi(x)-1)\, \nu(dx);
  \]
  this integral is finite due to the assumption on $\psi$ and H\"older's inequality.
  After subtracting the linear drift which is of order $O(T)$, the $Q$-distribution of $L$ therefore coincides with the $P$-distribution of $L'^{\leq\delta}$, which is obtained from $L'$ by truncating the jumps. As before we may assume that $L'=L'^{\leq\delta}$. To summarize, we have
  \begin{equation*}%
  E_Q[|L_T|]=E[|L'_T|]+O(T)
  \end{equation*}
  and hence it suffices to show that $E_Q[|L_T|]-E[|L_T|]=O(T)$.
  Indeed, $E_Q[|L_T|]=E[D_T|L_T|]$, H\"older's inequality and Lemma~\ref{le:LevyFacts}(vi),(x) yield%
  \begin{align*}
    \big|E_Q[|L_T|]-E[|L_T|]\big|
    &\leq E[|D_T-1||L_T|]\\
    &\leq \big\{E[(D_T-1)^2]\,E[L_T^2]\big\}^{1/2}\\
    &= \big\{\big(E[D_T^2]-1\big)\,E[L_T^2]\big\}^{1/2}\\
    &= \big\{\big(e^{T\br{Y,Y}_1}-1\big)\,T \br{L,L}_1\big\}^{1/2} =O(T),
  \end{align*}
  where we have used that both $Y$ and $L=L^{\leq\delta}$ are square-integrable.
\end{proof}

\begin{proof}[Proof of Theorem~\ref{thm:jumps}(i)]
  As in the proof of Theorem~\ref{th:callATM} we may assume that $S_0=0$ and then we have $E[|S_T|]=2E[S_T^+]$.

  \vspace{.5em}\noindent\emph{Step 1: $\alpha$-stable L\'evy case}. We first note the result for $Z$ and under the additional assumption that $f(x)=f_{-}1_{(-\infty,0)}(x)+f_+ 1_{(0,\infty)}(x)$ and that $\alpha_+=\alpha_-=:\alpha$. Then $Z$ is a centered $\alpha$-stable L\'evy motion with $\alpha \in (1,2)$ and in this case it is known that
  \begin{equation*}%
  E[|Z_T|]=C(\alpha,f_+,f_{-})\,T^{1/\alpha};
  \end{equation*}
  see Samorodnitsky and Taqqu~\cite[Property 1.2.17]{SamorodnitskyTaqqu.94}.%

  \vspace{.5em}\noindent\emph{Step 2: Stable-like L\'evy case with $\alpha_+=\alpha_-$}. We again consider $Z$ for the special case $\alpha_+=\alpha_-=:\alpha$ but now let $f$ be an arbitrary function satisfying  $f(x)-f_+=O(x)$ as $x \uparrow 0$ and $f(x)-f_-=O(x)$ as $x \downarrow 0$.
  For $\delta>0$ we define the function
  \[
    \psi(x):=1_{[-\delta,0]}(x)\bigg(\frac{f_{-}}{f(x)} + 1_{\{f_{-} = 0\}}\bigg)
           +1_{(0,\delta]}(x)\bigg(\frac{f_{+}}{f(x)} + 1_{\{f_{+} = 0\}}\bigg),
  \]
  where we use the convention $0/0:=0$. We have $|\psi(x)-1|=O(|x|)$. Indeed, consider the case $f_+>0$. Then for $x>0$ small enough we have that $f(x)\geq f_+/2>0$ and thus
  \[
    |\psi(x)-1|= \bigg|\frac{f_{+}-f(x)}{f(x)}\bigg| \leq \frac{2}{f_+} |f_{+}-f(x)| = O(x)\quad \mbox{as }x\downarrow 0.
  \]
  The case $f_+=0$ is trivial and $x<0$ is treated in the same way.
  As a result, by choosing $\delta$ small enough we can find a constant $M>0$ such that
  \[
    |\psi(x)-1|\leq M|x|,\quad x\in [-\delta,\delta].
  \]
  In particular, $\psi$ satisfies the integrability assumption of Lemma~\ref{lem:approxlevy}. In the sequel, we denote by $\nu$ the L\'evy measure of $Z$, i.e., $\nu(\cdot)=F(\kappa_0^{-1}(\cdot))$.

  \vspace{.5em}\noindent\emph{Case 2a: $f_+>0$ and $f_->0$.} In this case we have
  \[
    \nu'(dx):=\psi(x)\,\nu(dx) = \bigg(1_{[-\delta,0)}(x)\frac{f_{-}}{|x|^{1+\alpha}}+ 1_{(0,\delta]}\frac{f_+}{|x|^{1+\alpha}}(x)\bigg)\, dx.
  \]
  Note that $\nu'$ is the L\'evy measure of an $\alpha$-stable L\'evy motion whose jumps were truncated at magnitude $\delta$. As above, this truncation changes the option price only at the order $O(T)$. Now Lemma~\ref{lem:approxlevy} and Step~1 yield that
  \begin{equation}\label{eq:ProofAlphaLike}
    E[|Z_T|]=C(\alpha,f_+,f_{-})\, T^{1/\alpha} + O(T).
  \end{equation}

  \vspace{.5em}\noindent\emph{Case 2b: $f_+=0$ and $f_->0$.} Note that in this case the formula for $\nu'$ is not the desired one in general, since we have $f(x)$ instead of $f_+$ in
   \[
    \nu'(dx)=\psi(x)\,\nu(dx) =  \bigg(1_{[-\delta,0)}(x)\frac{f_{-}}{|x|^{1+\alpha}}+ 1_{(0,\delta]}\frac{f(x)}{|x|^{1+\alpha}}(x)\bigg)\, dx.
  \]
  However, the previous argument does apply if $f(x)=f_+=0$ for all $x>0$, and we shall reduce to this case.
  Indeed, Lemma~\ref{le:StableLikeTrivialLimit} stated below shows that the positive jumps of $Z$ are of finite variation, hence of integrable variation by~\eqref{eq:integrabilityZ}. By subtracting these jumps from $Z$ and compensating, we achieve that
  $f(x)=f_+=0$ for all $x>0$ and Theorem~\ref{th:finitevariation} shows that this manipulation affects the option price only at the order $O(T)$. Hence we may conclude as in Case~2a to obtain~\eqref{eq:ProofAlphaLike}.

  The case where $f_+>0$ and $f_-=0$ is analogous. Finally, if $f_+$ and $f_-$ both vanish, Lemma~\ref{le:StableLikeTrivialLimit} shows that
  $Z$ is of finite variation and Theorem~\ref{th:finitevariation} yields that the option price is of order $O(T)$.
  Hence~\eqref{eq:ProofAlphaLike} again holds since $C(\alpha,0,0)=0$.

  \vspace{.5em}\noindent\emph{Step 3: Stable-like L\'evy case with $\alpha_+\neq\alpha_-$}. We consider the case where $\alpha:=\alpha_+>\alpha_-$. The idea is to reduce to the case where $\alpha_+=\alpha_-$ but $f_-=0$; i.e., we get rid of all the negative jumps and show that this induces an error of order $o(T^{1/\alpha})$.

  Let $\mu$ be the random measure associated with the jumps of $Z$. By setting
  \[
    Z^+:=x^+ * (\mu-\nu(dx)\,dt)\quad\mbox{and}\quad Z^-:=x^- * (\mu-\nu(dx)\,dt)
  \]
  we decompose $Z=Z^+ +Z^-$, where $Z^+$ and $Z^-$ are L\'evy martingales having only positive and negative jumps, respectively, and L\'evy measures given by $\nu^+(dx)=1_{(0,\infty)}(x)\,\nu(dx)$ and $\nu^-(dx)=1_{(-\infty,0)}(x)\,\nu(dx)$. Note the abuse of notation: $x^+$ and $x^-$ refer to the positive and negative part of $x$ while $Z^\pm$ and $\nu^\pm$ are new symbols.

  We observe that $Z^+$ has $(\alpha_+,\alpha_+)$-stable-like jumps, where the corresponding function $f'$ in Definition~\ref{def:alphaLike} satisfies $f'(x)=0$ for $x<0$. In particular the left limit is $f'_-=0$. The martingale $Z^-$
  has analogous properties for $\alpha_-$.

  From Step~2 we know that $E[|Z^-_T|]=O(T^{1/\alpha_-})$ if $\alpha_-\in(1,2)$ and we have also seen that
  $E[|Z^-_T|]=O(T)$ if $\alpha_-\in(0,1)$ due to finite variation. For the remaining case
  $\alpha_-=1$ we have $E[|Z^-_T|]=O(T|\log T|)$ by Remark~\ref{rk:afterJumpsThm}(f).
  As $\alpha_-<\alpha_+$ and $\alpha_+>1$ we therefore have
  \[
    E[|Z^-_T|]=o(T^{1/\alpha_+})
  \]
  in all three cases for $\alpha_-$. On the other hand, we know from Step~2 applied to $\alpha_+$ that $E[|Z^+_T|]\sim CT^{1/\alpha_+}$. Therefore the leading order coefficient for $Z$ is the same as for $Z^+$, i.e.,
  \[
    E[|Z_T|] = C\big(\alpha_+,f_+,0\big)\,T^{1/\alpha_+}+o(T^{1/\alpha_+}).
  \]
  The case $\alpha_+<\alpha_-$ is analogous.

  \vspace{.5em}\noindent\emph{Step 4: General case}. Proposition~\ref{pr:approx}(ii) implies that
  \[
    E[|S_T-Z_T|]=o(T^{(1+\gamma)/\beta})
  \]
  and in particular $E[|S_T-Z_T|]=o(T^{1/\alpha})$. In view of the previous steps, this completes the proof.
\end{proof}

The following result was used in the preceding proof.

\begin{lemma}\label{le:StableLikeTrivialLimit}
  Let $L$ be a L\'evy process with $(\alpha_+,\alpha_-)$-stable-like small jumps for some $\alpha_+,\alpha_-\in(0,2)$. If the function $f$ from Definition~\ref{def:alphaLike} satisfies \mbox{$f_+=0$}, then the positive jumps of $L$ are of finite variation, that is,
  $\sum_{t\leq T} (\Delta L_t)^+ < \infty$ for any $T<\infty$.
\end{lemma}

\begin{proof}
  Let $\nu$ be the L\'evy measure of $L$. Since $f(x)=f(x)-f_+=O(x)$ as $x\downarrow 0$, there exist $\delta>0$ and $M>0$ such that $f(x)\leq Mx$ for $x\in (0,\delta]$. Therefore,
  \[
    \int_0^\delta x \,\nu(dx) = \int_0^\delta \frac{xf(x)}{|x|^{1+\alpha_+}} \,dx \leq M \int_0^\delta \frac{x^2}{|x|^{1+\alpha_+}} \,dx <\infty,
  \]
  showing that the small positive jumps are summable. Of course, the large jumps are always summable.
\end{proof}

We now come to the proof of the second part of the theorem. The main difference to the above is that
we cannot use the $1$-stable process as a reference since it is not integrable. Instead, we use the
normal inverse Gaussian process. It is the symmetry of its L\'evy density around zero that forces us to impose the condition $f_+=f_-$ in
the theorem to apply our method. Indeed, we are not aware of a suitable process with sufficiently asymmetric density for which the absolute moment asymptotics are known.

\begin{proof}[Proof of Theorem~\ref{thm:jumps}(ii)]
  The proof has the same structure as for part~(i).

  \noindent\emph{Step 1: Normal inverse Gaussian case}. First let $Z$ be a symmetric normal inverse Gaussian process with L\'evy measure
  \[
    \nu'(dx)=\frac{\rho}{\pi|x|}K_1(|x|),
  \]
  where $\rho>0$ and $K_\theta$ denotes the modified Bessel function of the third kind of order~$\theta$. Then we are in the setting of Theorem \ref{thm:jumps}(ii) with  $f_+=f_{-}=\rho/\pi$ by the properties of $K_1$; see, e.g., \cite[Formula~(9.6.9)]{AbramowitzStegun.64}. The absolute moments of $Z_t$ were calculated explicitly and for all $t$ by Barndorff-Nielsen and Stelzer \cite[Corollary~4]{BarndorffNielsenStelzer.05}. By their formula and another property of Bessel functions (see \cite[Formula~(9.6.8)]{AbramowitzStegun.64}) we have
  \[
    E[|Z_T|]=\frac{2\rho}{\pi}e^{\rho T} TK_0(\rho T) \sim \frac{2\rho}{\pi}\,T|\log T|=(f_+ +f_{-})\,T|\log T|.
  \]

  \vspace{.5em}\noindent\emph{Step 2: General L\'evy case}. Now let $Z$ be a $(1,1)$-stable-like L\'evy martingale with L\'evy measure $\nu(dx)=f(x)/|x|^2 dx$, where by assumption $f$ satisfies $f_0:=\lim_{x \to 0} f(x)=f_+=f_{-}$. As the case $f_0=0$ again follows from Lemma~\ref{le:StableLikeTrivialLimit}, we may assume that $f_0>0$. Then there exists a small $\delta>0$ such that $f$ is bounded away from zero on $[-\delta,\delta]$ and  we can define
  \[
    \psi(x):=1_{[-\delta,\delta]}(x)\frac{f_0 K_1(|x|)|x|}{f(x)}. %
  \]
  By choosing $\rho:=\pi f_0$ we have $\nu'(dx)=\psi(x)\,\nu(dx)$ on $[-\delta,\delta]$, where $\nu'$ is as in Step~1.
  As $f_0 K_1(|x|)|x|=f_0 +O(|x|)$ by \cite[Formula (9.6.11)]{AbramowitzStegun.64}, we have
  \[
    \psi(x)-1_{[-\delta,\delta]}(x)=1_{[-\delta,\delta]}(x) \frac{f_0 K_1(|x|)|x| -f(x)}{f(x)} = O(|x|)
  \]
  due to the assumption that $f(x)=f_0 +O(|x|)$.
  Making $\delta>0$ smaller if necessary, we conclude that $\int_{-\delta}^\delta (\psi(x)-1)^2\,\nu(dx)<\infty$
  and now the assertion follows from Lemma~\ref{lem:approxlevy} and Step~1. The last step to the general martingale case
  is as in the proof of part~(i).
\end{proof}

\appendix

\section{Appendix}

The following lemma collects some standard facts about L\'evy processes that are used throughout the text. A L\'evy process
$L$ is an adapted c\`adl\`ag process with independent and stationary increments and $L_0=0$.

\begin{lemma}\label{le:LevyFacts}
  Let $L$ be a L\'evy process having triplet $(b,c,\nu)$ with respect to the truncation function $h(x)=x1_{|x|\leq 1}$.
  \begin{enumerate}[topsep=3pt, partopsep=0pt, itemsep=1pt,parsep=2pt]
    \item $\int_{\R} 1\wedge|x|^2\,\nu(dx)<\infty$.
    \item There is a decomposition $L_t=L^a_t + L^b_t+L^c_t+At$ into independent L\'evy processes such that $L^a$ is a compound Poisson process, $L^b$ is a purely discontinuous martingale with bounded jumps, $L^c=\sqrt{c}W$ is a scaled Brownian motion and  $A\in\R$.
    \item Let $p\in [1,\infty)$. Then $\int_{|x|>1} |x|^p\,\nu(dx)<\infty$ if and only if $E[|L_t|^p]<\infty$ for all $t\geq0$.
    In particular, if $L$ has bounded jumps, or equivalently if the support of $\nu$ is compact, then $E[|L_t|^p]<\infty$ for all $p\in[1,\infty)$.
    \item $L$ is a martingale if and only if the two conditions $\int_{|x|>1} |x|\,\nu(dx)<\infty$ and $b+\int_{|x|>1} x\,\nu(dx)=0$ hold.
    \item If $L$ is integrable, then $E[L_t]=tE[L_1]$ and $L_t-tE[L_1]$ is a martingale.
        \item If $L$ is a square-integrable martingale, then
         \[
          E[L_t^2] = E[[L,L]_t]=\br{L,L}_t=t\br{L,L}_1=tc+t\textstyle{\int_{\R}} |x|^2\,\nu(dx)<\infty.
         \]
    \item If $c=0$, $L$ is of finite variation if and only if $\int_{|x|\leq1} |x|\,\nu(dx)<\infty$ and of integrable variation
          if and only if $\int_{\R} |x|\,\nu(dx)<\infty$. In that case the total variation $\Var(L)_t:=\int_0^t|dL_s|$ is a L\'evy process satisfying
          $E[\Var(L)_1]= \int |x|\,\nu(dx) + |\int x\,\nu(dx)| $.
    \item For any $\delta>0$ there is a decomposition $L=L^{\leq\delta}+L^{>\delta}$ into two independent L\'evy processes satisfying $|\Delta L^{\leq\delta}|\leq\delta$ as well as $|\Delta L^{>\delta}|>\delta$ on the set $\{|\Delta L^{>\delta}|\}>0$. The corresponding L\'evy measures are given by $\nu^{\leq\delta}(dx)=1_{[-\delta,\delta]}(x)\,\nu(dx)$ and $\nu^{>\delta}(dx)=1_{\R\setminus[-\delta,\delta]}(x)\,\nu(dx)$. If $L^{>\delta}$ has no Brownian component, it is a compound Poisson process with drift.
    \item If $L$ is a martingale, the stochastic exponential $\cE(L)$ is again a martingale.
    \item If $L$ is a square-integrable martingale, then so is $\cE(L)$ and moreover
          $E[\cE(L)^2_t]=\exp(t\br{L,L}_1)$.
  \end{enumerate}
\end{lemma}

\begin{proof}
  (i)--(viii) can be found in any advanced textbook about L\'evy processes; see, e.g.,~\cite{Sato.99}. Statement~(ix) is
  \cite[Proposition 8.23]{ContTankov.04}. One way to deduce the formula in~(x) is to
  use Yor's formula in
 \begin{align*}
   \cE(L)_t^2
     & = \cE\big(2L + [L,L] - \br{L,L} + \br{L,L}\big)_t \\
     & = \cE\big(2L + [L,L] - \br{L,L}\big)_t \exp(t\br{L,L}_1).
 \end{align*}
 Noting that $2L + [L,L] - \br{L,L}$ is a L\'evy martingale, (ix) yields the result.
\end{proof}

The following makes precise a remark from the introduction.

\begin{remark}\label{re:representation}
  Let $S$ be any c\`adl\`ag martingale with absolutely continuous predictable characteristics, then $S$ can be represented in the form~\eqref{eq:canonRepresentationIto}. Indeed, let
  \begin{equation*}%
    dB_t=b_t\,dt,\quad dC_t=\sigma_t^2\,dt,\quad d\nu_t=K_t(dx)\,dt
  \end{equation*}
  be the characteristics of $S$ with respect to the trivial truncation function $h(x)=x$ (cf.\ \cite[Chapter~II]{JacodShiryaev.03} for background). The latter choice is possible since $S$ is a martingale, which then implies $B=0$. Moreover, let $F$ be any atomless $\sigma$-finite measure on $\R$ such that $F(\R)=\infty$. Then there exist a Brownian motion $W$ and
  a Poisson random measure $N$ with compensator $F(dx)dt$ such that~\eqref{eq:canonRepresentationIto} holds. Moreover, $\kappa$ and $K$ satisfy
  the relation $K_t(A)=F(\kappa_t^{-1}(A))$ for any Borel set $A$, where $\kappa_t^{-1}$ denotes the preimage with respect to the spatial variable $x$. To be precise, the construction of $W$ and $N$ may necessitate an enlargement of the probability space, but this is harmless since we are interested only in distributional properties. We refer
  to Jacod~\cite[Theorem 14.68(a)]{Jacod.79} for further details.
\end{remark}
\bibliographystyle{plain}
\newcommand{\dummy}[1]{}

\end{document}